\documentclass[11pt]{article}

\newcommand{\eat}[1]{}

\usepackage{latexsym}
\usepackage{amsmath,amsthm}

\newtheorem{theorem}{Theorem}

\newtheorem{corollary}{Corollary}

\usepackage{booktabs}

\usepackage{tabulary}

\numberwithin{equation}{section}
\numberwithin{theorem}{section}

\RequirePackage{fix-cm}
\usepackage{breqn}
\usepackage{mdwtab}
\usepackage[ruled,lined]{algorithm2e}
\usepackage{setspace}
\usepackage{graphicx}
\usepackage{multirow}
\usepackage{array}
\newcolumntype{B}[1]{>{\raggedright\let\newline\\\arraybackslash\hspace{0pt}}m{
#1}}

\newcolumntype{L}[1]{>{\raggedright\let\newline\\\arraybackslash\hspace{0pt}}m{
#1}}
\newcolumntype{C}[1]{>{\centering\let\newline\\\arraybackslash\hspace{0pt}}m{#1}
}
\newcolumntype{R}[1]{>{\raggedleft\let\newline\\\arraybackslash\hspace{0pt}}m{#1
}}

\begin{document}

\title{Scheduling Under Power and Energy Constraints}

\author{Mohammed Haroon Dupty \hspace{0.5in} Pragati
  Agrawal \hspace{0.5in} Shrisha Rao \\
{\tt mohammed.haroon@iiitb.org} \hspace{0.1in} {\tt pragati.a.in@ieee.org} \hspace{0.1in} {\tt shrao@ieee.org}%
}
\date{}
\maketitle

\begin{abstract}
Given a system model where machines have distinct speeds and power
ratings but are otherwise compatible, we consider various problems of
scheduling under resource constraints on the system which place the
restriction that not all machines can be run at once.  These can be
power, energy, or makespan constraints on the system.  Given such
constraints, there are problems with divisible as well as
non-divisible jobs. In the setting where there is a constraint on
power, we show that the problem of minimizing makespan for a set of
divisible jobs is $\mathcal{NP}$-hard by reduction to the knapsack
problem. We then show that scheduling to minimize energy with power
constraints is also $\mathcal{NP}$-hard.  We then consider scheduling
with energy and makespan constraints with divisible jobs and show that
these can be solved in polynomial time, and the problems with
non-divisible jobs are $\mathcal{NP}$-hard. We give exact and
approximation algorithms for these problems respectively.
\end{abstract}


\maketitle

\section{Introduction}
\label{sec:intro}

It is highly desirable that every system of machines be as efficient
as can be with regard to the resources that it consumes. This also
requires appropriate usage of machines within a larger system, in
addition to proper design and operation of individual machines for
good solo performance.  Of particular concern is the appropriate
scheduling of jobs over a system of machines.  Scheduling of jobs can
be considered effective when it meets stated goals such as minimal
makespan (time to completion), minimal energy consumed, and minimal
cost incurred.  Thus, it is necessary to attempt to create effective
scheduling algorithms in line with such desirable goals.

Additionally, however, it is also seen in practice that many if not
most systems also operate under constraints that may be externally or
internally imposed.  These may be power, energy, or makespan
constraints on the system.  Systems connected to smart power grids can
provide good examples for power-constrained scheduling problems. Smart
grids and usage characteristics such as demand-side
management~\cite{brennan2010} or demand response~\cite{SpeesLave2007}
necessitate that systems control their power usage to certain limits
lest there be penalties or other unwelcome consequences.  Such
concerns with power usage are of course prevalent in conventional
industrial systems~\cite{jessoe2013}, but also in large computing
systems~\cite{Ranganathan2010,akshay2013}.  Additionally, many
contemporary systems are powered by multiple sources, which almost
always have varying associated energy capacities, as in cases of
renewable sources like solar~\cite{lavania2011} and
wind~\cite{Apt2010}. Many performance optimization
problems~\cite{alenawy2005} for real-time systems rely on a fixed
energy budget during an operation.

Similarly, other problems where there are hard deadlines, may require
energy-minimal scheduling while respecting constraints on makespan.
Therefore, it is necessary to consider the problems of how scheduling
algorithms can be designed to be effective in terms of optimizing
multiple quality parameters such as makespan and energy consumption,
while running on systems subject to different kinds of resource
constraints.

We consider these problems in the present paper, with a model of
resource constrained systems whose machines are similar in their
capabilities but may have different working and idle power ratings and
working speeds.  Based on different types of machines cooperatively
running similar jobs, we classify and analyze the problems of resource
constrained offline scheduling with non-identical interconnected
machines and independent jobs.  Fundamentally, the jobs run on the
machines can themselves be either \emph{divisible}, i.e., can be
broken up in arbitrary ways into chunks of whatever sizes one may
wish, and \emph{non-divisible}, in which case the jobs cannot be
subdivided.  A little insight suggests that the problems of scheduling
non-divisible jobs are at least as hard as those of scheduling
divisible jobs.

The basic problem of minimizing the \emph{makespan} (time to
completion of a set of jobs) under a power constraint is seen
(Section~\ref{sec:makespan_power}) to be
$\mathcal{NP}$-hard~\cite{Garey1990} even for divisible jobs, by
reduction to the \emph{knapsack problem}~\cite{martelli}, a canonical
problem in complexity theory.  This implies that minimal-makespan
scheduling is hard also for non-divisible jobs.

Next, we consider power constrained energy-minimal scheduling, which
is known \cite{pagrawal2014} to be at least as hard as
minimum-makespan scheduling.  Quite obviously, this too turns out to
be $\mathcal{NP}$-hard (Section~\ref{sec:energy_power}).

We then consider the problem of scheduling to minimize the makespan
given energy budget. Unlike other power-constrained scheduling
problems, this problem turns out to be solvable in polynomial time for
divisible jobs.  For non-divisible jobs, the scheduling problem to
minimize makespan is known to be $\mathcal{NP}$-hard even without any
constraints.  So with energy constraints, it should be
$\mathcal{NP}$-hard as well for non-divisible jobs
(Section~\ref{sec:makespan_energy}).

Last, we consider the problem of energy minimal scheduling under
makespan constraint.  Like the problem of minimal-makespan scheduling
under an energy constraint, this too turns out to be solvable in
polynomial time for divisible jobs and $\mathcal{NP}$-hard for
non-divisible jobs (Section~\ref{sec:energy_makespan}).

These results suggest that many other interesting problems of
effective scheduling under pairs of various resource constraints are
computationally intractable, i.e., unlikely to have optimal general
solutions that can be computed efficiently.  However, while we are
essentially presenting negative results, we believe that doing so
opens up the possibility for focused approaches to finding
approximation algorithms and other solutions to such problems.

We give approximation algorithms for divisible jobs given a power
constraint: for minimizing makespan
(Section~\ref{sec:makespan_power}), and for minimizing energy
consumption (Section~\ref{sec:energy_power}).  Each algorithm is seen
to have a bound of ($1+\epsilon)OPT$.  We then give exact algorithms
for minimal makespan and energy scheduling under respective
constraints of energy and makespan for divisible jobs; and
approximation algorithms for non-divisible jobs
(Section~\ref{sec:makespan_energy} and
Section~\ref{sec:energy_makespan}).

Our basic model is general (not domain-specific) and therefore
extensible for more specific classes of systems.  Our analyses with
respect to power and energy can also be translated for systems that
consume any other resources (e.g., water, fuel), and are subject to
similar constraints from sources supplying the same.

General literature on
scheduling~\cite{Herrmann2006,pinedo2012scheduling} has considered a
multitude of objectives related to time, but has not considered energy
and power in the general setting.  It is well known that energy costs
are the primary costs in some large systems such as data
centers~\cite{datacenter2,datacenter1}, and that the energy consumed
by idle machines in such systems is a significant part of their
overall energy consumption, so the lack is significant.  Even the
literature on scheduling for meeting more than one goal, which is
called multi-objective
scheduling~\cite{drozdowski2014,lee2014,shi2012}, does not address
these issues.  Some domain-specific literature~\cite{ZomayaLee2012}
focuses on energy optimality, but calls on particular features and
technologies of those domains (such as DVFS), and does not consider
the limits on scheduling due to power, energy or makespan.

In the next section we formally introduce our system model along with
the notations used in this paper.

\section{System Model} \label{sec:sm}

Consider $m$ machines $c_1$ to $c_m$ forming a system $\mathcal{C}$.
The working power of machine $c_i$ is denoted as $\mu(c_i)$, and the
idle power as $\gamma(c_i)$.  The sum of the idle powers of all the
machines is given by $\Gamma$.  The speed of machine $c_i$ is denoted
as $\upsilon(c_i)$.  The speed (throughput of work per unit time) of a
machine is fixed throughout its working tenure, and all machines
process identical jobs (so that any job can be run on any machine).
 
When $\upsilon(c_i) = 1$, with $1 \leq i \leq m$, then we can say that
the machines are identical in their working capacities or speeds,
implying that they can execute and complete any job given to them in
equal time.  And if in that case all jobs are executed sequentially
on one machine, then that time taken is $W$.

All machines stay on for the duration of the makespan of the whole set
of jobs.  (This is reasonable considering that in many systems the
cycle time to stop/restart a machine is large; however, it is not
a restrictive assumption, as stopping idle machines is equivalent to 
the idle power of those machines being zero.)

All machines in the system work in parallel, and the maximum working
time of the system to execute a given set of jobs is $T$, which is the
makespan of the system for that set of jobs.  All jobs are
independent, meaning a job need not wait for completion of any other
particular job to start its execution.  This implies that
\begin{equation}
 T = \max {t_i}, \, \forall i, 1 \leq i \leq m
 \label{label:tmax}
\end{equation}

If machine $c_i$ works only for time $t_i$, then the idle time of
machine $c_i$ is given by $T-t_i$.  The amount of work done by machine
$c_i$ is represented by $w(c_i)$.
\begin{equation}
 w(c_i) = t_i \upsilon(c_i)
\label{label:w}
\end{equation}

The sum of the work done by all machines is equal to the total work to
be done, i.e., $\sum_{i=1}^m w(c_i) = W$. $E$ represents the energy
consumption of the complete system.

Since we consider the energy consumption in working as well as idle
states, the energy consumed by machine $c_i$ is the sum of the energy
consumed in the working state and that consumed in the idle state.
The energy consumed by machine $c_i$ in the working state is given by
$\mu(c_i){t_i}$, and the energy consumed in the idle state by
$\gamma(c_i)(T-t_i)$, so the total energy consumed by $c_i$ is given
by $\mu(c_i)\tau(c_i) + \gamma(c_i)(\kappa(c_i))$.

Overall, for the entire system $\mathcal{C}$, we get, after
simplification:
\begin{equation}
 E = \sum_{i = 1}^{m} \left[\frac{w(c_i)}{\upsilon(c_i)}(\mu(c_i) - 
\gamma(c_i)) + 
\gamma(c_i) T\right]
\label{label:energy}
\end{equation}
where $1 \leq i \leq m$. 

The system $\mathcal{C}$ can have constraints like power budget $
{P}$, , energy budget $ {E}$ and makespan budget $ {T}$. $ {P}$ and $
{E}$ are the maximum power and energy that can be drawn altogether by
all the machines of the system.  $ {T}$ is the maximum time the system
is on.

Evidently, we have to assume that $ {P} > \Gamma + \min(\mu(c_i) -
\gamma(c_i))$, i.e., the power is enough to keep all the machines idle
and let at least one machine run a job.  It is likewise necessary to
assume that $ {P} < \sum_i^m \mu(c_i)$, i.e., that the power budget is
not sufficient to run all machines at once. (For it is sufficient,
then the availability of power is no longer a constraint on
scheduling). Similar assumptions are to be made for other constraints
also.

In considering the types of jobs to be executed by the system of
machines, the simplest kind are divisible jobs, which can be divided
in arbitrary ways, with it being possible to run chunks of any size on
any machine~\cite{bharadwaj2003}.  This is obviously something of an
abstraction, but is satisfied to an extent in practice with such tasks
as pumping water.  The other type of jobs are non-divisible jobs,
which come in fixed-size chunks that cannot be divided in arbitrary
ways.

It is of interest to note that scheduling of non-divisible jobs is
certainly no easier than the scheduling of divisible jobs; if a set of
non-divisible jobs can be scheduled effectively (by whatever measure
of effectiveness one may choose to apply in a certain context), then
an equivalent set of divisible jobs can also be scheduled effectively
by dividing them into the same chunk sizes as those of the
non-divisible jobs.  Therefore, results about divisible jobs set a
baseline of difficulty for all jobs; if a certain class of problems is
intractable when the jobs are divisible, it is also that way with
non-divisible jobs~\cite{pinedo2012scheduling}.

We formulate the following problems:
\begin{enumerate}
\item Scheduling to minimize makespan given a constraint on power 
(Section~\ref{sec:makespan_power}).
\item Scheduling to minimize energy given a constraint on power 
(Section~\ref{sec:energy_power}).
\item Scheduling to minimize makespan given a constraint on energy 
(Section~\ref{sec:makespan_energy}).
\item Scheduling to minimize energy given a constraint on makespan 
(Section~\ref{sec:energy_makespan}).
\end{enumerate}

We explore the first two problems with divisible jobs only since these
are $\mathcal{NP}$-hard even for divisible jobs. For the next two
problems we consider solutions for both divisible and non-divisible
jobs.  A short summary of the results is given below along with bounds
for the algorithms.

%\begin{table}[!t]
\begin{table*}[htbp]
\renewcommand{\arraystretch}{1.5}
\caption{Summary of results.}
\label{res}
\centering
\scalebox{1}{
\begin{tabular}{p{3.8cm}  p{1.5cm}  p{1.51cm}  p{1.8cm} p{1.53cm} p{1.187cm} }
\toprule
 \textbf{Problem}  & \textbf{Type of Job} & \textbf{Hardness} & 
\textbf{Type of Algorithm} & \textbf{Time Complexity} & \textbf{Bound} \\ 
\midrule
 Minimize makespan given constraint on power & Divisible & 
$\mathcal{NP}$-hard & 
Approximate 
& $\mathcal{O}(m^2 \lfloor \frac{m}{\epsilon}\rfloor) $ & (1$+ \epsilon$ ) \\
\midrule
 Minimize energy given constraint on power & Divisible & 
$\mathcal{NP}$-hard & 
Approximate & 
$\mathcal{O}(m\log m)$ & 2 \\
\midrule
 \multirow{2}{*}{\parbox {4cm}{Minimize makespan given constraint on 
energy}} & 
Divisible & Polynomial time & Exact & $\mathcal{O}(m\log m)$ & -\\
\cline{2-6}
 & Non-divisible & $\mathcal{NP}$-hard & Approximate & 
$\mathcal{O}(m^2n)$ & $\frac{19}{12}+ \epsilon$ 
\\
\midrule
 \multirow{2}{*}{\parbox{4cm}{Minimize energy given constraint on 
makespan}} & 
Divisible & Polynomial time & Exact & $\mathcal{O}(m\log m)$ & -\\
\cline{2-6}
 & Non-divisible & $\mathcal{NP}$-hard & Approximate & 
$\mathcal{O}(mn)$ & 
$1+\frac{\eta_{max}}{\eta_{min}}$ \\
\bottomrule
\end{tabular}}
\end{table*}

\section{Scheduling to Minimize Makespan Under Limited Power} 
\label{sec:makespan_power}

%\subsection{Minimizing Makespan} \label{min_makespan}

The first problem we consider is to minimize the makespan (time to
completion of a set of jobs) of system $\mathcal{C}$, while under a
power constraint.

\begin{theorem} \label{theorem1}
If the power budget for a system of machines is limited, minimizing
the makespan for a set of divisible jobs is an $\mathcal{NP}$-hard problem.
\end{theorem}

\begin{proof}
Obviously, to minimize the makespan, we need to choose a subset of
machines from $\mathcal{C}$ to achieve the highest cumulative speed,
given the power constraint.  This amounts to finding some $r < m$
indices $i_r$ (the indices of the machines chosen to run) such that
\(\sum^r \upsilon(c_i)\) is the greatest possible.  In other words, we
wish to maximize \(\sum^r \upsilon(c_{i_r})\), subject to the constraint:
\begin{equation}
 \sum^r (\mu(c_{i_r}) - \gamma(c_{i_r})) + \Gamma \leq  {P}
\end{equation}
$\Gamma$ can be moved to the right so that $ {P} - \Gamma$ is
treated as one constant, say $Z$, and \(\mu(c_{i_r}) -
\gamma(c_{i_r})\) can be likewise treated as one variable, say
$d_{i_r}$.  Then the problem reduces to: maximize \(\sum
\upsilon(c_{i_r})\) subject to:
\begin{equation}
\sum^r d_{i_r} < Z 
\end{equation}
which is an instance of the knapsack problem~\cite{Garey1990}, a
canonical $\mathcal{NP}$-hard problem.
\end{proof}

It may be noted that this is a somewhat surprising result, because
minimum-makespan scheduling of divisible jobs in a system without a
power constraint is trivial---one has to just run all machines.

Given the lack of ease, as previously discussed, of scheduling
non-divisible jobs, the following obtains.

\begin{corollary} \label{corollary1}
If the power budget for a system of machines is limited, minimizing
the makespan for a set of non-divisible jobs is an $\mathcal{NP}$-hard
problem.
\end{corollary}

Therefore, we may generally say that given a limited power budget,
scheduling jobs on a system is an $\mathcal{NP}$-hard problem.

%\subsection{Approximation Algorithm} \label{approx_makespan_power}
\subsection*{Approximation Algorithm}

We give the following approximation algorithm for the objective of
minimizing makespan.  Without loss of generality, we assume the
working and idle power consumptions, and the speeds of machines to be
integers.

As seen in the previous section, the problem of minimizing makespan,
given a constraint on power can be formulated as a knapsack problem.
Here we minimize $T$ or maximize $\frac{1}{T}$.
\begin{equation}
\begin{split}
\text{maximize} & {\sum}^r \upsilon(c_{i_r})\\
\text{subject to}
&\sum^r (\mu(c_{i_r}) - \gamma(c_{i_r}))  \leq  {P} - \Gamma
%\label{eq:}
\end{split}
\end{equation}

Considering $\upsilon(c_{i})$ as profit for component $i$ (machine
$c_i$) and $(\mu(c_{i}) - \gamma(c_{i}))$ as weight of component $i$
(machine $c_i$), we get the classical knapsack formulation. Hence we
can apply approximation algorithm from the large set of algorithms
developed for knapsack problem. \eat{ We present an approximation
algorithm~\ref{makespan_power_algo} which is easy to implement and has
worst-case bound of $\frac{1}{2}$ when maximizing $\frac{1}{T}$, which
comes to a bound of 2 when maximizing $T$.

\begin{algorithm}[htbp]
\LinesNumbered
\SetAlgoLined
\SetKwInOut{Input}{input}\SetKwInOut{Output}{output}
\Input{Set of machines ($\mathcal{C}$), number of machines ($m$), working power 
of machines ($\mu(c_i)$), idle 
power of machines ($\gamma(c_i)$), sum of idle 
power of all the machines ($\Gamma$), speed of machines ($\upsilon(c_i)$),
power constraint ($ {P}$)}
\Output{Working set ($\mathcal{R}$)}
\BlankLine
\For{$i = 1$ to $m$}{
calculate $\frac{\upsilon(c_i)}{\mu(c_i) - \gamma(c_i)}$
}
max-sort $(\frac{\upsilon(c_i)}{\mu(c_i) - \gamma(c_i)})$\;
$\alpha \leftarrow \Gamma$\;
$\mathcal{R} \leftarrow \emptyset$\;
$\mathcal{A} \leftarrow \mathcal{C}$\;
\For{$i = 1$ to $m$}{
\eIf{$\alpha <  {P}$}
{$margin \leftarrow  {P} -\alpha$\;
%$\forall c_i \in \mathcal{A}$\;
   $p_i \leftarrow \mu(c_i) - \gamma(c_i)$\;
   \eIf{ $p_i \leq margin$}{
   $\mathcal{R} \leftarrow \mathcal{R} + \{c_i\}$\;
   }{
   %\eIF{}{}{}   
   \textbf{if} $(p_i \leq \alpha) \wedge (\sum_{i \in \mathcal{R}} 
\upsilon(c_{i-1}) < \upsilon(c_i) )$ \textbf{then} \\
      $\mathcal{R} \leftarrow {c_i}$\;
      
}
$\mathcal{A} \leftarrow \mathcal{C} - \mathcal{R}$\;
   $\alpha \leftarrow \Gamma + \sum_{i \in \mathcal{R}}(\mu(c_i) - 
\gamma(c_i))$\;
 }{stop\;}
}
  working set $\leftarrow \mathcal{R}$\;
\caption{Approximation algorithm for reducing makespan subject to power 
constraint with divisible jobs}
\label{makespan_power_algo}
\end{algorithm}

The time complexity of Algorithm~\ref{makespan_power_algo} is
$\mathcal{O}(m\log{}m)$. 

In Algorithm~\ref{makespan_power_algo}, we sort the machines in
decreasing order of their profit by weight ratio.  Then we check the
power consumption of machines sequentially: if the power of some
machine(s) is (are) within the constraint then, we add said machine(s)
to our working set $\mathcal{R}$.  But if the power consumption of a
machine is not within the margin, then we compare the profit (sum of
speeds) of machines in current working set with the marginal profit of
this particular machine.  If the sum of speeds of machines in
$\mathcal{R}$ is more than the speed of particular machine, then we
continue with the same set, else we choose only that particular
machine and update $\mathcal{R}$.

\begin{theorem} \label{theorem2}
The worst case bound of Algorithm~\ref{makespan_power_algo} to maximize
$\frac{1}{T}$ is $\frac{1}{2}$, so the makespan achieved is at most
twice the optimal.
\end{theorem}

\begin{proof}
Since we are taking machines in working set according to their profit
per unit weight, the optimal profit can be simply attained if
\begin{equation}
\sum^r (\mu(c_{i_r}) - \gamma(c_{i_r}))  =  {P} - \Gamma
\end{equation}
For the worst case, suppose there are only two machines in the
system. Taking $\upsilon(c_{1}) = (\mu(c_{1}) - \gamma(c_{1})) =
\upsilon(c_{2}) = (\mu(c_{2}) - \gamma(c_{2})) = k$ and
$ {P} - \Gamma = 2k - 1$.  As the bound is given by current 
performance divided by optimal performance. The lower bound $L$, is given by,
\begin{equation}
L = \frac{k}{2k - 1}
\end{equation}
 $L \approx \frac{1}{2}$.
\end{proof}

Hence the worst-case bound on makespan comes to $2$.}
%%%%%%%%%%%%%%%%%%%%%%%%%%%%%%%%%%%%%%%%%%%%%%%%%%%%%%%%%%%%%%%%%%%%%%%%%%%%%%%%
%%%%%%%%%%%%%%%%%%%%%%%%%%%%%%%%%%%%%%%%%%%%%%%%%%%%%%%%%%%%%%%%
\eat{
\subsection{Fully Polynomial Time Approximation Scheme (FPTAS) for Minimizing 
Makespan}
 Here we have to minimize $ {T}$ or maximize $\frac{1}{T}$.
\begin{equation}
\begin{split}
\begin{array}{ll@{}ll}
\text{maximize} &{\sum}^r \upsilon(c_{i_r})\\
\text{subject to}
&\sum^r (\mu(c_{i_r}) - \gamma(c_{i_r}))  \leq  {P} - \Gamma
\end{array}
\end{split}
\end{equation}
}
We use a suitably modified version of the well known fully polynomial
time approximation scheme (FPTAS) for the knapsack
problem~\cite{Ibarra1975}.  Let $V$ be the speed of the fastest
machine, i.e., $V = \max_i \upsilon(c_i)$.  Then $mV$ is a trivial
upper-bound on the speed that can be achieved by any solution.  For
each $i \in \{1, \ldots, m\}$ and $\upsilon \in \{ 1, \ldots, mV\}$,
let $S_{i,\upsilon}$ denote a subset of $\{c_1, \ldots, c_i\}$ whose
total speed is exactly $\upsilon$ and whose total power requirement is
minimized. Let $A(i,\upsilon)$ denote the total power requirement of
the set $S_{i,\upsilon}$ ($A(i,\upsilon) = \infty$ if no such set
exists). Clearly $A(1, \upsilon)$ is known for every $ \upsilon \in \{
1, \ldots, mV\}$. The following recurrence helps compute all values
$A(i,\upsilon)$ in $\mathcal{O}(m^2V)$ time.
\begin{equation*}
A(i+1,\upsilon) =
\end{equation*}
\footnotesize
\begin{equation*}
    \begin{cases}
        \min\{A(i, \upsilon), (\mu(c_{i+1})-\gamma(c_{i+1}))+A(i, 
\upsilon-\upsilon(c_{i+1}))\}, & 
\text{if $\upsilon(c_{i+1}) < \upsilon$ }\\
       A(i+1, \upsilon)=A(i, \upsilon), & \text{otherwise}
        \end{cases}
\end{equation*}

\normalsize
The maximum speed achievable by the machines with total power bounded
by P is $\max\{\upsilon| A(m, \upsilon)\leq P\}$. We thus get a
pseudo-polynomial algorithm for minimizing makespan under power
constraints.

If the speeds of machines were small numbers, i.e., they were bounded
by a polynomial in $m$, then this would be a regular polynomial time
algorithm. To obtain a FPTAS we ignore a certain number of least
significant bits of speeds of machines (depending on the error
parameter $\epsilon$), so that the modified speeds can be viewed as
numbers bounded by a polynomial in $m$ and $\frac{1}{\epsilon}$. This
will enable us to find a solution whose speed is at least
($1-\epsilon)OPT$ in time bounded by a polynomial in $m$ and
$\frac{1}{\epsilon}$.

\begin{algorithm}[htbp]
\LinesNumbered
\SetAlgoLined
\SetKwInOut{Input}{input}\SetKwInOut{Output}{output}
\Input{Set of machines ($\mathcal{C}$), number of machines ($m$), working power 
of machines ($\mu(c_i)$), idle 
power of machines ($\gamma(c_i)$), sum of idle 
power of all the machines ($\Gamma$), speed of machines ($\upsilon(c_i)$),
power constraint ($ {P}$)}
\Output{Working set ($\mathcal{R}$)}
$V \leftarrow \max_{\forall i \epsilon \{1\dot{...}m\}} \upsilon(c_i)$\\
Given $\epsilon > 0$, let $K$ = $\frac{\epsilon V}{m}$.\\
For each machine $c_i$ define $\upsilon'(c_i) = \lfloor 
\frac{\upsilon(c_i)}{K}\rfloor$.\\
With these as speeds of machines, using the dynamic programming algorithm, find 
the set of machines with maximum total speed, say $\mathcal{R}$.\\
Return $\mathcal{R}$. 
\caption{FPTAS for makespan with divisible jobs}
\label{makespan_power_algo_fptas}
\end{algorithm}
%%%%%%%%%%%%%%%%%%%%%%%%%%%%%%%%%%%%%%%%%%%%%%%%%%%%%%%%%%%%%%%%%%%%%%%%%%%%%%%%
%%%%%%%%%%%%%%%%%%%%%%%%%%%%%%%%%%%%%%%%%%%%%%%%%%%%%%%%%%%%%%%%

\begin{theorem}\label{theorem_3}
If set $\mathcal{R}$ is output by the 
algorithm~\ref{makespan_power_algo_fptas} and $\upsilon(\mathcal{R})$
denotes ${\sum}_{\forall i \in \mathcal{R}} \upsilon(c_{i})$  then,
$\upsilon(\mathcal{R}) \geq (1-\epsilon)OPT$.
\end{theorem}

\begin{proof}
Let ${O}$ denote the optimal set. For any machine $c_i$, because of rounding 
down, 
$K\upsilon(c_i)$ can be smaller than $\upsilon(c_i)$, but by not more than $K$. 
Therefore,
\begin{equation*}
\upsilon(O)-K\upsilon(O)\leq mK.
\end{equation*}
The dynamic programming step must return a set at least as good as $O$ under 
the 
new profits. Therefore, 
\begin{equation*}
\upsilon(\mathcal{R})\geq K\upsilon(O)\geq \upsilon(O)-mK = 
OPT-\epsilon V \geq (1-\epsilon)OPT, 
\end{equation*}
where the last inequality follows from the observation that $OPT\geq V$. It 
directly follows that minimum makespan $T \leq (1+\epsilon)OPT$.
\end{proof}

\section{Scheduling to Minimize Energy Under Limited Power} 
\label{sec:energy_power}

Minimizing the makespan for a set of jobs running on a system does
not guarantee that the energy consumed is minimized.  To see why, we
can compare the problems as done in prior work~\cite{pagrawal2014}.

If the idle power of every machine is equal to its working power,
i.e., if $\mu(c_i) = \gamma(c_i)$, we have the following
from~\eqref{label:energy}:
\begin{align}
  E &= \sum_{i = 1}^{m} [\mu(c_i)\tau(c_i) + \gamma(c_i)(T-\tau(c_i))] \notag \\
    &= \sum_{i = 1}^{m} \mu(c_i)T 
\end{align}

Thus, in this particular setting, the energy is minimized if the
makespan $T$ is minimized.  However, more generally, the idle power of
machines may be less than their working power, so that $\gamma(c_i)=
z_i \cdot \mu(c_i)$, where $0 \leq z_i \leq 1$.  In this case, we
get:
\begin{align}      
  E &= \sum_{i = 1}^{m} [\mu(c_i)\tau(c_i) + 
\gamma(c_i)( {T}-\tau(c_i))] \notag \\       
    &= \sum_{i = 1}^{m} [\mu(c_i)\tau(c_i) + 
z_i\mu(c_i)( {T}-\tau(c_i))] \notag \\    
    &= \sum_{i = 1}^{m} \mu(c_i)[\tau(c_i) + z_i {T}-z_i\tau(c_i)]   
\end{align}

This in turn simplifies to:
\begin{equation} \label{label:powerratio}
E = \sum_{i = 1}^{m} \mu(c_i)[z_iT +\tau(c_i)(1 - z_i)]
\end{equation}

Therefore, in this case, even minimizing $T$ does not minimize $E$,
and the problem of energy-minimal scheduling is always at least as
hard as that of minimize makespan.  This gives us the following.

\begin{theorem} \label{theorem3}
Minimal-energy scheduling of either divisible or non-divisible jobs,
given a power constraint, is $\mathcal{NP}$-hard.
\end{theorem}

Even more simply, if it were possible to efficiently compute
minimal-energy schedules, it would be possible to minimize the
makespan simply by using an energy-minimal schedule with $\mu(c_i) =
\gamma(c_i)$, which contradicts Result~\ref{theorem1}.

This too is a surprising result in a way, because it is
known~\cite{pagrawal2015} that energy-minimal scheduling of divisible
jobs in a system without a power constraint can be achieved in
\(\mathcal{O}(m)\), i.e., linear time.

\subsection*{Approximation Algorithm}

The problem of minimizing energy subject to constraint on power is
harder than minimizing for makespan.  To state the problem in more
general form, minimizing energy ($E$) can be easily seen as maximizing
its inverse, i.e., $\frac{1}{E}$.  Hence the problem can be formally
written as:
\begin{equation*}
\text{maximize} \frac{\sum^r \upsilon(c_{i_r})}{\sum^r (\mu(c_{i_r}) - 
\gamma(c_{i_r})) + 
\Gamma}
%\label{eq:}
\end{equation*} 
subject to constraint
\begin{equation}
\sum^r (\mu(c_{i_r}) - \gamma(c_{i_r}))  \leq  {P} - \Gamma
%\label{eq:}
\end{equation} 

Here we cannot find profit per unit weight for each element, as the
objective function is not a linear function of just one property of
elements of the set.  When we need to minimize $E$, we arrange
machines in an order such that the first machine is the one with
smallest \(\frac{\mu(c_i)-\gamma(c_i)+\Gamma}{\upsilon(c_i)}\) and
afterwards in non-decreasing order of
$\frac{\mu(c_i)-\gamma(c_i)}{\upsilon(c_i)}$, where, \(1\leq i \leq
m\).  Since, here we need to maximize $1/E$ within power constraint,
we need to arrange machines in such an order that we give a machine
more priority if it reduces more energy consumption of system per unit
increase of power consumption.  Hence, we arrange machines in an order
such that the first machine has highest
\(\frac{\upsilon(c_i)}{(\mu(c_i)-\gamma(c_i)+\Gamma)(\mu(c_i)-\gamma(c_i))}\)
and afterwards in non-increasing order of
$\frac{\upsilon(c_i)}{(\mu(c_i)-\gamma(c_i))^2}$, where, \(1\leq i
\leq m\).

Algorithm~\ref{energy_power_algo} for minimizing energy is based on this
ordering.

\begin{algorithm}[htbp]
\LinesNumbered
\SetAlgoLined
\small
\SetKwInOut{Input}{input}\SetKwInOut{Output}{output}
\Input{Set of machines ($\mathcal{C}$), number of machines ($m$), working power 
of machines ($\mu(c_i)$), idle 
power of machines ($\gamma(c_i)$), sum of idle 
power of all the machines ($\Gamma$), speed of machines ($\upsilon(c_i)$),
power constraint ($ {P}$)}
\Output{Working set ($\mathcal{R}$)}
\BlankLine

\For{$i = 1$ to $m$}{
calculate 
\(\frac{\upsilon(c_i)}{(\mu(c_i)-\gamma(c_i)+\Gamma)(\mu(c_i)-\gamma(c_i))}\) 
}
$(\frac{\upsilon(c_1)}{(\mu(c_1)-\gamma(c_1)+\Gamma)(\mu(c_1)-\gamma(c_1))})
\leftarrow 
max(\frac{\upsilon(c_i)}{(\mu(c_i)-\gamma(c_i)+\Gamma)(\mu(c_i)-\gamma(c_i))})$ 
\;
\For{$i = 2$ to $m$}{
calculate $\frac{\upsilon(c_i)}{(\mu(c_i) - \gamma(c_i))^2}$
}
max-sort $(\frac{\upsilon(c_i)}{(\mu(c_i) - \gamma(c_i))^2})$\;
$\alpha \leftarrow \Gamma$\;
$\mathcal{R} \leftarrow \emptyset$\;
$\mathcal{A} \leftarrow \mathcal{C}$\;
%$e_0 = e_1$\;
\For{$i = 1$ to $m$}{
\eIf{$\alpha <  {P}$}
{$margin \leftarrow  {P} -\alpha$\;
%$\forall c_i \in \mathcal{A}$\;
   $p_i \leftarrow \mu(c_i) - \gamma(c_i)$\;
   $e_i \leftarrow \frac{\mu(c_i) - \gamma(c_i)}{\upsilon(c_i)}$\;
   $ce \leftarrow \frac{\Gamma + \sum_{i \in \mathcal{R}}(\mu(c_i) - 
\gamma(c_i))}{\sum_{i \in \mathcal{R}} \upsilon(c_i)}$\;
   
   \eIf{ ($p_i \leq margin) \wedge (ce \leq e_i)$}{
   $\mathcal{R} \leftarrow \mathcal{R} + \{c_i\}$\;
    }
   {
    %\eIF{}{}{}   
   \textbf{if} $(p_i \leq \alpha) \wedge (\sum_{i \in \mathcal{R}} 
\frac{\upsilon(c_{i-1})}{\mu(c_{i-1}) - \gamma(c_{i-1})} 
<\frac{\upsilon(c_i)}{\mu(c_i) - \gamma(c_i)})$ \textbf{then} \\
      $\mathcal{R} \leftarrow {c_i}$\;
      }
      $\mathcal{A} \leftarrow \mathcal{C} - \mathcal{R}$\;
   $\alpha \leftarrow \Gamma + \sum_{i \in \mathcal{R}}(\mu(c_i) - 
\gamma(c_i))$\;
 }{stop\;}
}
  working set $\leftarrow \mathcal{R}$\;
\caption{Approximation algorithm for reducing energy subject to power 
constraint with divisible jobs}
\label{energy_power_algo}
\end{algorithm}

The time complexity of Algorithm~\ref{energy_power_algo} is also
$\mathcal{O}(m\log{}m)$.  Here the first machine chosen is the one
which has
maximum\\ \(\frac{\upsilon(c_i)}{(\mu(c_i)-\gamma(c_i)+\Gamma)(\mu(c_i)-\gamma(c_i))}\).
Later machines are arranged in decreasing order of their profit to
weight ratio.  The profit is
$\frac{\upsilon(c_i)}{\mu(c_i)-\gamma(c_i)}$, and weight remains same
as previous algorithm, i.e., $\mu(c_i)-\gamma(c_i)$.  Now we check the
power consumption of machines sequentially; if the power consumption
of a particular machine is less than or equal to the margin (the power
remaining after already-scheduled machines draw their power
requirements) then we compare the energy of the current machine with
the current energy. When both conditions are satisfied, we add this
current machine to our working set $\mathcal{R}$.  If the power
consumption of a particular machine is not within the limit, then we
compare the profit of this machine and the sum of profit of machines
in working set. If the profit of this machine is more, then we choose
only this particular machine and update $\mathcal{R}$ else we continue
with same set $\mathcal{R}$.

\begin{theorem} \label{theorem5}
The worst case bound for Algorithm~\ref{energy_power_algo} to maximize
$\frac{1}{E}$ is $\frac{1}{2}$, so the energy consumed is at most
twice the optimal.
\end{theorem}

\begin{proof}
For the worst case, suppose there are only two machines in the
system.  Taking $\frac{\upsilon(c_{1})}{(\mu(c_{1}) - \gamma(c_{1}))} =
(\mu(c_{1}) - \gamma(c_{1})) = \frac{\upsilon(c_{2})}{(\mu(c_{2}) -
  \gamma(c_{2}))} = (\mu(c_{2}) - \gamma(c_{2})) = k$ and
$ {P} - \Gamma = 2k - 1$.  As the bound is given by current 
performance divided
by optimal performance, we have,
\begin{equation}
L = \frac{\frac{\upsilon(c_{1})}{\Gamma + (\mu(c_{1}) 
-\gamma(c_{1}))}}{\frac{\upsilon(c_{1}) + \upsilon(c_{2})}{\Gamma + \mu(c_{1}) 
- \gamma(c_{1}) + \mu(c_{2}) - \gamma(c_{2})}}
\label{eq:lbe}
\end{equation}
Solving~\eqref{eq:lbe} we get:

\begin{align}
L &= \frac{(\Gamma + 2(\mu(c_{1}) - \gamma(c_{1})))}{2(\Gamma+ (\mu(c_{1}) - 
\gamma(c_{1})))} \notag \\
%\label{eq:}
  &= 1- \frac{\Gamma}{2(\Gamma +  \mu(c_{1}) - \gamma(c_{1})}
\end{align}

Clearly the value of the bound depends upon the power consumption
specifications of the system. The bound depends on the total idle
power consumption and value of the difference between working and idle
power specification of the machine for which
\(\frac{\upsilon(c_i)}{(\mu(c_i)-\gamma(c_i)+\Gamma)(\mu(c_i)-\gamma(c_i))}\)
is highest and whose inclusion does not violate power constraint.
Since idle power will always be less than working power, $L \approx 1
- \frac{1}{2} = \frac{1}{2}$ in the worst case, so the worst case
bound for minimization of energy is 2. %\qedhere
\end{proof}

%%%%%%%%%%%%%%%%%%%%%%%%%%%%%%%%%%%%%%%%%%%%%%%%%%%%%%%%%%%%%%%%%%%%%%%%%%%

%%%%%%%%%%%%%%%%%%%%%%%%%%%%%%%%%%%%%%%%%%%%%%%%%%%%%%%%%%%%%%%%%%%%%%%%%%%

\section{Minimizing Makespan Given Energy Budget} 
\label{sec:makespan_energy}

We analyse the problem of minimizing makespan given energy constraint with 
divisible and non-divisible jobs.
\subsection{Divisible Jobs}

Contrary to the problems with power constraint, the problem of
minimizing makespan, given an energy budget with divisible job can be
formulated as a fractional knapsack problem. This is because unlike
power fractional amount of energy can be given to a machine to run for
some part of the makespan of the system. Here we minimize $T$ or
maximize $\frac{1}{T}$.
\begin{equation*}
\text{maximize}{\sum}^r \upsilon(c_{i_r})
%\label{eq:}
\end{equation*}
subject to constraint
\begin{equation}
\sum^r (\mu (c_{i_r}) - \gamma(c_{i_r})) t_{i_r} + \Gamma  {T} \leq 
 {E} 
%\label{eq:}
\end{equation}

Since $\upsilon(c_{i})$ gives work done per unit time and power rating
gives energy required per unit time for the machine to operate, we can
take their ratio as a measure of efficiency. This will be a measure of
amount of work done per unit energy consumed. Based on this parameter
we give an algorithm to get the minimum makespan and the set of
machines to be used.

\begin{algorithm}[htbp]
\LinesNumbered
\SetAlgoLined
\SetKwInOut{Input}{input}\SetKwInOut{Output}{output}
\Input{Set of machines ($\mathcal{C}$), number of machines ($m$), working power 
of machines ($\mu(c_i)$), idle 
power of machines ($\gamma(c_i)$), sum of idle 
power of all the machines ($\Gamma$), speed of machines ($\upsilon(c_i)$),
energy constraint ($ {E_c}$), total work ($W$)}
\Output{Working set ($\mathcal{R}$), Makespan ($ {T}$)}
\BlankLine
\For{$i = 1$ to $m$}{
calculate $\frac{\upsilon(c_i)}{\mu(c_i) - \gamma(c_i)}$
}
max-sort $(\frac{\upsilon(c_i)}{\mu(c_i) - \gamma(c_i)})$\;

%$\alpha \leftarrow \Gamma$\;%
$\mathcal{R} \leftarrow \emptyset$\;
$T_o \leftarrow \frac{W}{\upsilon(c_1)}$\;
$e \leftarrow [(\mu(c_1) - \gamma(c_1)) + \Gamma] T$\\
\For{$i = 1$ to $m$}{
$T \leftarrow \frac{T}{i}$\\
$e_{prev} \leftarrow e$\\
$e \leftarrow \sum_{ j = 1}^{i}[(\mu(c_j) - \gamma(c_j)) + \Gamma] T $\\
\eIf{$e <  {E}$}{
$r \leftarrow r + 1$\;
}{
$margin \leftarrow  {E} - e_{prev}$\\
$\psi(p_i) \leftarrow margin * (\frac{\upsilon(c_i)}{\mu(c_i) - \gamma(c_i)})$\\
%$ {W} \leftarrow  {W} - \psi(p_i) $\\
assign $\psi(p_i)$ amount of work to machine $c_i$.\\
break\;
}
}
%$ {T} \leftarrow \frac{ {W}}{\upsilon_1}$\\
$ {T} \leftarrow \frac{ {T_o}}{r}$\\
$ {R} \leftarrow  {R} \cup_{j = 1}^{r + 1} c_j$

Return $ {T}$\\
Return $ {R} \cup \{c_{r + 1}\}$

\caption{Exact algorithm for minimizing makespan given energy budget with 
divisible jobs}
\label{makespan_energy_algo}
\end{algorithm}

Algorithm~\ref{makespan_energy_algo} takes Set of machines, working
and idle power of machines, speed of machines, energy constraint (E)
and total work as input and gives the minimum makespan and the subset
of machines to be used as output. It assumes energy is not sufficient
to complete all work and finds the maximum amount of work that can
scheduled. The machines are sorted in the order of their
efficiency. Now, every machine is given work which the machine can
complete in makespan $ {T}$. So the machines will be active for an
equal amount of time.  This makespan is iteratively decreased based on
the number of machines added to the working set. When the energy
constraint is violated at an addition of a particular machine, that
machine is given the fractional amount of energy available so that all
of available energy is used. The time complexity of
Algorithm~\ref{makespan_energy_algo} turns out to be
$\mathcal{O}(m\log{}m)$.

%%%%%%%%%%%%%%%%%%%%%%%%%%%%%%%%%%%%%%%%%%%%%%%%%%%%%%%%%%%%%%%%%%%%%%%%%%%

\subsection{Non-divisible Jobs}

%\subsection{Minimizing Makespan given Energy budget with indivisible job} 
\label{min_makespan_given_energy_algo_indivisible_job}

Since even simple scheduling in multiple parallel machines with
indivisible jobs itself is $\mathcal{NP}$-Hard, the problem of
scheduling in parallel machines given energy constraint is certainly
harder. For this we give an approximation algorithm designed in
similar terms as previous algorithm but with indivisible jobs.  Here
we find maximal set of machines which can work within the Energy
constraint. For this we sort the machines in terms of efficiency and
add them one by one to the working set if total energy requirement is
within the constraint. We stop at the machine where energy requirement
is violated. Now we give maximum amount of work to this machine to fit
within the remaining energy margin. The remaining set of jobs are
assigned to the previously chosen machines by Longest Processing
Time(LPT) algorithm. In LPT algorithm the jobs are first sorted in
decreasing order of their size and are assigned to the least loaded
machine one by one to minimize the makespan.

\begin{algorithm}[htbp]
%begin{spacing}{0.5}
%linespread{0.01}
%\tiny
%\scriptsize
\footnotesize
\LinesNumbered
\SetAlgoLined
\SetKwInOut{Input}{input}\SetKwInOut{Output}{output}
\Input{Set of machines ($\mathcal{C}$), number of machines ($m$), working power 
of machines ($\mu(c_i)$), idle 
power of machines ($\gamma(c_i)$), sum of idle 
power of all the machines ($\Gamma$), speed of machines ($\upsilon(c_i)$),
energy constraint ($ {E}$), set of jobs ($\mathcal{P}$), weight of jobs 
($\psi(p_i)$)}
\Output{Working set ($\mathcal{R}$), makespan ($ {T}$)}
\BlankLine
\For{$i = 1$ to $m$}{
calculate $\frac{\upsilon(c_i)}{\mu(c_i) - \gamma(c_i)}$
}
max-sort $(\frac{\upsilon(c_i)}{\mu(c_i) - \gamma(c_i)})$; and max-sort 
$\psi(p_i)$\;
$ {R} \leftarrow \emptyset $;
$ {W} \leftarrow \sum_{i=1}^{n}\psi(p_i)$;
$t(1) \leftarrow \frac{ {W}}{\upsilon(c_1)}$;
$max_t \leftarrow 1$\;
$e \leftarrow [(\mu(c_1) - \gamma(c_1)) + \Gamma] T$\\
\For{$i = 1$ to $m$}{
$r \leftarrow i$\\
\For{$j = 1$ to $n$}{
$k \leftarrow argmin_{\forall l \in \{1 \ldots r\}} t(l)$\\
$t(k) \leftarrow t(k) + \frac{\psi(p_j)}{\upsilon(c_k)}$\\
\If{$(t(k) > t(max_t))$}{
$max_t \leftarrow k$\;
}
}
$ {T} \leftarrow t(max_t)$\\
$e_{prev} \leftarrow e$\\
$e \leftarrow \sum_{ i = 1}^{r}(\mu(c_j) - \gamma(c_j))t(i) + \Gamma T $\\
\If{$e >  {E}$}{
$margin \leftarrow  {E} - e_{prev}$\\
find max subset of $\{\psi(p_i) \mid \forall i = 1 \ldots n\}$ to fit in energy 
margin with speed of machine $c_r$\\
assign all jobs in the subset, to machine $c_r$\\
$r \leftarrow r - 1$\\
break\;
}
}
$T \leftarrow 0$\\
\For{$j = 1$ to $r$}{
$k \leftarrow argmin_{\forall l \epsilon \{1 \ldots r\}} t(l)$\\
$t(k) \leftarrow t(k) + \frac{\psi(p_j)}{\upsilon(k)}$\\
assign the job with weight $\psi(p_j)$ to machine $c_k$\\
$ {R} \leftarrow  {R} \cup c_j$\\
\If{$(t(k) > T)$}{
$ {T} \leftarrow t(k)$\;
}
}
Return $ {R} \cup \{c_{r + 1}\}$;
Return $ {T}$;
%end{spacing}
\caption{Approximation algorithm for reducing makespan given energy budget 
with indivisible jobs}
\label{makespan_energy_algo_indivisible}
\end{algorithm}

The time complexity of Algorithm~\ref{makespan_energy_algo_indivisible} is
$\mathcal{O}(m^2n)$.

\begin{theorem} \label{theorem6}
The worst case bound of Algorithm~\ref{makespan_energy_algo_indivisible} is 
$(\frac{19}{12} + \epsilon)$. So the worst case makespan ${T}$ for a set of 
parallel machines given Energy constraint with non-divisible jobs is 
$(\frac{19}{12} + \epsilon)OPT$. 
\end{theorem}

\begin{proof}
The problem is same as the previous one but with indivisible jobs.
When we sort machines in terms of their efficiencies, the last machine
which is added to the working set will be given only the marginal
energy left. So it can't operate for full makespan. Ideally makespan
will be least if we have exact energy to be distributed among $r$ set
of machines proportionately such that all machines will be active for
the same time. But since this is not always the case, the marginal
energy left after selecting $r-1$ machines is allotted to machine
$r$. Since time during which machine $c_r$ is active is less than the
makespan of the system, the Optimal solution will give maximum amount
of work possible within the energy margin available to the machine
$r$.  Algorithm~\ref{makespan_energy_algo_indivisible} uses the well
known PTAS for subset sum problem to take the maximum subset of jobs
that can be assigned to the machine $r$ . So it will return
$(1-\epsilon)OPT$ subset size to be assigned to the machine $r$.
These extra jobs that optimal would have assigned to machine $r$ need
to be accommodated in the $r-1$ machines. In worst case this can
increase the makespan of the system by $\epsilon$ amount.  For
scheduling in the $r-1$ machines
Algorithm~\ref{makespan_energy_algo_indivisible} uses Longest
Processing Time (LPT) algorithm. \cite{Dobson:1984:SIT:2054.2059}
provided the bound of ($\frac{19}{12})OPT$ for LPT algorithm applied
to a system of parallel machines with different speeds. With this
bound, the worst case bound for
Algorithm~\ref{makespan_energy_algo_indivisible} can go up to
($\frac{19}{12} + \epsilon)OPT$.
\end{proof} 
%%%%%%%%%%%%%%%%%%%%%%%%%%%%%%%%%%%%%%%%%%%%%%%%%%%%%%%%%%%%%%%%%%%%%%%%%%%
\section{Minimizing Energy Given Limit on Makespan} \label{sec:energy_makespan}

Here, given a constraint on makespan $ {T}$, the problem is to find
the optimal subset of machines so that energy consumed is
minimized. We divide this problem too between divisible and
non-divisible jobs. We show that the problem of minimizing energy
given a limit on the makespan with divisible jobs is polynomially
solvable whereas with non-divisible jobs it is $\mathcal{NP}$-Hard. We
provide exact and approximation algorithms for these problems
respectively.

\subsection{Divisible Jobs}
Like energy, amount of time given for a particular machine is also
divisible, i.e., the constraint parameter can be taken in
fractions. Hence, it can be formulated in terms of fractional knapsack
problem.

Energy of the system is given by
\begin{equation}
 {E} = \sum^r (\mu(c_{i_r}) - \gamma(c_{i_r})) t_{i_r} + \Gamma 
 {T}\\
\end{equation}
For energy minimality, all machines should be given equal amount of 
time\cite{pagrawal2015}. Then
\begin{align}
 {E} &= \sum^r [(\mu(c_{i_r}) - \gamma(c_{i_r})) + \Gamma] 
 {T} \notag \\
 {E} &= \sum^r [(\mu(c_{i_r}) - \gamma(c_{i_r})) + 
\Gamma]\left(\frac{ {w(c_i)}}{\upsilon_{i_r}}\right)
\end{align}
Hence the problem can be written as
\begin{equation*}
\text{minimize} {\sum^r [(\upsilon(c_{i_r}) - \gamma(c_{i_r})) + \Gamma] 
\left(\frac{ w(c_i)}{\upsilon_{i_r}}\right)}\\
\end{equation*} 
subject to constraint
\begin{equation}
\left(\frac{ {w(c_i)}}{\upsilon_i}\right) \leq T
\end{equation} 

We can see that if the jobs are non-divisible this problem becomes
$\mathcal{NP}$-Hard otherwise it is not, i.e., since jobs can be
arbitrary divisible, this LP formulation is perfectly solvable. If it
is non-divisible then, it becomes an integer linear program which is
not polynomially solvable.

\begin{algorithm}[htbp]
\LinesNumbered
\SetAlgoLined
\SetKwInOut{Input}{input}\SetKwInOut{Output}{output}
\Input{Set of machines ($\mathcal{C}$), number of machines ($m$), working power 
of machines ($\mu(c_i)$), idle 
power of machines ($\gamma(c_i)$), sum of idle 
power of all the machines ($\Gamma$), speed of machines ($\upsilon(c_i)$),
Makespan limit ($ {T}$)}
\Output{Working set ($\mathcal{R}$)}
\BlankLine
\For{$i = 1$ to $m$}{
calculate $\frac{\upsilon(c_i)}{\mu(c_i) - \gamma(c_i)}$
}
max-sort $(\frac{\upsilon(c_i)}{\mu(c_i) - \gamma(c_i)})$\;
$\mathcal{R} \leftarrow \emptyset$\;
$ {W} \leftarrow \text{total work}$\\
\For{$i = 1$ to $m$}{
\eIf{$(( {W} \geq 0 ) \wedge (i \leq m))$}{ 
$w(c_i) \leftarrow \upsilon(c_i) *  {T}$\\
$ {W_{prev}} \leftarrow  {W}$\\
$ {W} \leftarrow  {W} - w(c_i)$\\
$\mathcal{R} \leftarrow  \mathcal{R} \cup {c_i}$\\
\If{$( {W} < 0)$}{
$t_i \leftarrow \frac{W_{prev}}{\upsilon(c_i)}$\\
assign $ {W}_{prev}$ amount of job to machine $c_i$ \\
}
}{
STOP
}
}

Working Set $\leftarrow \mathcal{R}$

\caption{Exact algorithm for minimizing energy given limit on makespan with 
divisible 
jobs}
\label{energy_makespan_algo}
\end{algorithm}
Algorithm~\ref{energy_makespan_algo} solves this problem in
$\mathcal{O}(m)$ time. It first sorts the machines based on their
efficiencies as described in the previous algorithms. Then it iterates
from the most energy efficient machine to least efficient machine
assigning maximum job that the machine can complete within the
makespan limit. The last machine in the subset of machines chosen is
provided with work job that can be finished before makespan for energy
minimality.

%%%%%%%%%%%%%%%%%%%%%%%%%%%%%%%%%%%%%%%%%%%%%%%%%%%%%%%%%%%%%%%%%%%%%%%%%%%

\subsection{Non-divisible Jobs} 
\label{min_energy_given_makespan_indivisible_algo}

For non-divisible jobs, the problem formulation becomes an Integer
Program which is not solvable in polynomial time. Hence we try to get
an approximation algorithm for the same.

As the makespan is fixed, all machines in the working set should be
active for the maximum amount of time within this constraint. The
optimal algorithm will include least number of most efficient machines
from the set of machines to the working set. This implies that any
approximation algorithm will include more number of machines than the
optimal.  If we sort the machines based on their efficiencies and try
to include least number of these machines to our working set, we see
that this problem reduces to a variant of variable size bin packing
problem (VSBP) which is a well studied $\mathcal{NP}$-Hard problem.

There are many variants of the standard one dimensional bin packing
problem.  The most common is given $n $ items with sizes $a_1, \ldots,
a_n \in (0,1]$, find a packing in unit-sized bins that minimizes the
number of bins used. Here all the bins are assumed to be of same
size. The variant of this problem is when bins of different sizes
are allowed. Since, in our problem of scheduling to minimize energy
given a limit on makespan, each machine is assumed to have different
speeds, the time taken to complete a job in one machine may be
different than the time taken by other machines to complete the same
job. Hence our problem resembles variable size bin packing problem.
Hence we model our approximation algorithm based on VSBP problem.

\begin{algorithm}[htbp]
\LinesNumbered
\SetAlgoLined
\SetKwInOut{Input}{input}\SetKwInOut{Output}{output}
\Input{Set of machines ($\mathcal{C}$), number of machines ($m$), working power 
of machines ($\mu(c_i)$), idle 
power of machines ($\gamma(c_i)$), sum of idle 
power of all the machines ($\Gamma$), speed of machines ($\upsilon(c_i)$),
makespan limit ($ {T}$), weight of jobs ($\psi(p_j)$)}
\Output{Working set ($\mathcal{R}$), energy ($ {E}$)}
\BlankLine
\For{$i = 1$ to $m$}{
calculate $\frac{\upsilon(c_i)}{\mu(c_i) - \gamma(c_i)}$\\
$t_i \leftarrow 0$
}
max-sort $(\frac{\upsilon(c_i)}{\mu(c_i) - \gamma(c_i)})$\;
max-sort $\psi(p_j)$\\
\For{$j = 1$ to $n$}{
\For{$i = 1$ to $m$}{
\If{$(\frac{\psi(p_j)}{\upsilon(c_i)} \leq  {T} - t_i)$}{
assign $\psi(p_j)$ to $c_i$\\
$t_i \leftarrow t_i + \frac{\psi(c_j)}{\upsilon(c_i)}$\\
$\mathcal{R} \leftarrow \mathcal{R} \cup c_i$\\
$break$\;
}
}
}
$ {E} \leftarrow \sum_{i \epsilon R}[(\mu(c_i) - \gamma(c_i))t_i + 
\Gamma  {T}] $\\
Return $ {E}$\\
Return $\mathcal{R}$
\caption{Approximation algorithm for reducing energy given limit on makespan 
with indivisible jobs}
\label{energy_makespan_algo_indivisible}
\end{algorithm}

Algorithm~\ref{energy_makespan_algo_indivisible} employs a variant of
first fit decreasing algorithm but without sorting the machines based
on their order of speed. In FFD algorithm, both bins and jobs are
sorted based on their sizes.  Here jobs are sorted in non-increasing
order of their sizes but machines are sorted based on their
efficiencies instead of speeds, which is required for energy
minimality. Because of this the bound on
Algorithm~\ref{energy_makespan_algo_indivisible} is not as small as
FFD for VSBP. Once sorted, a job is assigned to the most efficient
machine within the current working set which can accommodate it within
the time constraint. When a job cannot be finished within the makespan
limit by any machine in the working set, the next machine is activated
and included in the working set. The time complexity of
Algorithm~\ref{energy_makespan_algo_indivisible} is $\mathcal{O}(mn)$.

\begin{theorem} \label{theorem7}
The worst case bound on the number of machines selected by 
Algorithm~\ref{energy_makespan_algo_indivisible} is $2.OPT$.
\end{theorem}

\begin{proof}

This is on the similar lines as with the first-fit decreasing
algorithm\cite{Friesen:1986}.  But we get a higher bound because of
not sorting machines based on their speeds.

Let $k$ be the number of machines chosen by the by the algorithm and
let $k^*$ be the optimal number of machines. We note that the jobs
have been sorted in non-increasing order. Let $S$ be the number of
most energy efficient machines which can complete all jobs such that
all machines are active for the complete makespan. So we have the
trivial bound $k^* \geq S$. Let $b \leq k$ be an arbitrary machine in
the working set of
algorithm~\ref{energy_makespan_algo_indivisible}. We will analyze the
following two cases: $b$ is assigned a job which will take
$>\frac{T}{2}$ time on the fastest machine of the system or it is not
assigned a job of such size. If such a job takes half the time of
makespan on the fastest machine, then surely it will take more time on
other machines. Suppose $b$ is assigned a job $i$ which will take
$\frac{T}{2}$ time on the fastest machine. Then the previously
considered jobs $i'< i$ all will take $>\frac{T}{2}$ time in any
machine and each machine $b' < b$ must be assigned one of these jobs,
so we have $\geq b$ jobs of size $\frac{T_{\upsilon_{max}}}{2}$. No
two of these jobs can be assigned to the same machine in any
assignment, so optimal uses at least b machines, i.e., $k* \geq b$.

Suppose $b$ is not assigned a job of size
$>\frac{T_{\upsilon_{max}}}{2}$. Then no used machine $b'' > b$ is
assigned an item of size $>\frac{T_{\upsilon_{max}}}{2}$.  But each of
these machines must be assigned atleast one job to be included in the
working set. So the $k-b$ machines $b, b + 1, \ldots, k-1$ together
are assigned $\geq (k - b)$ jobs. We know that none of these jobs
could have been assigned to any machine $b' < b$. We consider two
sub-cases. If $b \leq (k-b)$, then we can imagine assigning to every
machine $b' < b$ one of these ($k-b$) jobs, which would give us $b-1$
machines taking time more than the makespan constraint $T$. This
implies that $S > b-1$. On the other hand, if $b >$ ($k-b$), then we
can imagine assigning each of the ($k-b$) jobs to a different machine
$b' < b$, giving us $(k-b)$ machines taking more than $T$ time. Then
$S >$ ($k-b$). So for any $b$ we have in all cases that either $k^*$
$\geq b$ or $k^*$ $\geq (k - b)$. Now we choose $b$ so that it will
maximize the minimum of $b$ and $(k - b):$ Equating $b = (k-b)$ gives
$b = \frac{1}{2}k$, and we take $b = \lceil{\frac{1}{2}k}\rceil$ to
get an integer. Then we have that $k^* \geq \lceil \frac{1}{2}k\rceil
\geq \frac{1}{2}k$, or $k^* \geq (k - \lceil \frac{1}{2}k \rceil) \geq
\frac{1}{2}k$.Hence we have $k \leq 2k^*$.
\end{proof}

\begin{theorem} \label{theorem8}
The worst case bound on energy $ {E}$ consumed by the working set of 
parallel machines for non-divisible jobs with 
Algorithm~\ref{energy_makespan_algo_indivisible} is given by\\ $\Big(1 + 
\frac{\eta_{max}}{\eta_{min}}\Big) {E_{OPT}}$ where $\eta(c_i)$ is 
$\frac{\upsilon(c_i)}{\mu(c_i)-\gamma(c_i)}$.
\end{theorem}

\begin{proof}
Energy consumed by the system of $r$ parallel machines is given by 
\begin{equation}
 {E} = \sum_{i=1}^r [(\mu(c_i)-\gamma(c_i))]t(c_i) + \Gamma T
\end{equation}
which can be rewritten as
\begin{equation}
 {E} = \sum_{i=1}^r 
[(\mu(c_i)-\gamma(c_i))]\frac{w(c_i)}{\upsilon(c_i)} + \Gamma T
\end{equation}
Let $E$ represent the energy consumed by
Algorithm~\ref{energy_makespan_algo_indivisible} and $ {E}_{OPT}$ be
the optimal energy consumption of the system. Also, let $r_o$ denote
the optimal number of machines and $r$' denote the number of machines
selected by Algorithm~\ref{energy_makespan_algo_indivisible}. We have
seen that $r'=2r_o$ in the worst case.  Then
%\begin{eqnarray}
\begin{equation}      
%\begin{split}
\frac{ {E}}{ {E_{OPT}}} = 
\frac{\sum_{i=1}^{2r_o} 
[(\mu(c_i)-\gamma(c_i))]\frac{w(c_i)}{\upsilon(c_i)} + \Gamma 
T}{\sum_{i=1}^{r_o} [(\mu(c_i)-\gamma(c_i))]\frac{w_o(c_i)}{\upsilon(c_i)} + 
\Gamma T}\\
\end{equation}

\begin{equation}
\begin{split}
\frac{ {E}}{ {E_{OPT}}} =&
\frac{\sum_{i=1}^{r_o} 
[(\mu(c_i)-\gamma(c_i))]\frac{w(c_i)}{\upsilon(c_i)} + \Gamma 
T}{\sum_{i=1}^{r_o} [(\mu(c_i)-\gamma(c_i))]\frac{w_o(c_i)}{\upsilon(c_i)} + 
\Gamma T}\\ 
&+ \frac{\sum_{i=r_o+1}^{2r_o}[(\mu(c_i)-\gamma(c_i))]\frac{w(c_i)}{\upsilon(c_i)}
}{\sum_{i=1}^{r_o} [(\mu(c_i)-\gamma(c_i))]\frac{w_o(c_i)}{\upsilon(c_i)} + 
\Gamma T}
\end{split}
\end{equation}

The first part of the equation covers the first $r_o$ machines both in
the numerator and the denominator. So, the power and speed ratings are
same. Since $w_o(c_i)$ will be greater than $w(c_i)$ in the numerator,
as optimal algorithm will assign all the jobs within the first $r_o$
machines, and the approximation algorithm will fail to assign all the
jobs to those $r_o$ machines and will need extra machines to schedule
these jobs. Hence the first part of the equation can be upper bounded
by 1.
\begin{equation}
\frac{ {E}}{ {E_{OPT}}} =
1 
+\frac{\sum_{i=r_o+1}^{2r_o}[(\mu(c_i)-\gamma(c_i))]\frac{w(c_i)}{\upsilon(c_i)}
}{\sum_{i=1}^{r_o} [(\mu(c_i)-\gamma(c_i))]\frac{w_o(c_i)}{\upsilon(c_i)} + 
\Gamma T}
\end{equation} 
Let $\eta(c_i)$ denote $\frac{\upsilon(c_i)}{\mu(c_i)-\gamma(c_i)}$ which is a 
measure of efficiency of the machine $c_i$, i.e., it is a measure of the 
amount of energy converted to useful work by the machine. Then
\begin{equation}
\frac{ {E}}{ {E_{OPT}}} = 1 
+\frac{\sum_{i=r_o+1}^{2r_o}(\frac{1}{\eta(c_i)})w(c_i)}{\sum_{i=1}^{r_o} 
(\frac{1}{\eta(c_i)})w_o(c_i) + \Gamma T}
\end{equation} 
Since $\Gamma  {T}$ is positive and to upper bound the ratio, we can 
ignore this term in the denominator.
\begin{equation}
\frac{ {E}}{ {E_{OPT}}} \leq 1 
+\frac{\sum_{i=r_o+1}^{2r_o}(\frac{1}{\eta(c_i)})w(c_i)}{\sum_{i=1}^{r_o} 
(\frac{1}{\eta(c_i)})w_o(c_i)}\\
\end{equation}
Because the machines are sorted, $\eta$ is increasing from $1$ to $m$.
\begin{equation}
\frac{ {E}}{ {E_{OPT}}} = 1 
+\frac{(\frac{1}{\eta_{min}})\sum_{i=r_o+1}^{2r_o}w(c_i)}{(\frac{1}{\eta_{max}}
)\sum_{i=1}^{r_o} w_o(c_i)}\\
\end{equation}
Now we are bothered about total size of work assigned to\\ 
${i=(r_o+1)} \ldots {2r_o}$ machines. If we divide the ratio by $r_o$ 
machines, we 
get the average size of work assigned by the 
Algorithm~\ref{energy_makespan_algo_indivisible} to ${i=(r_o+1)} \ldots {2r_o}$ 
machines and the optimal average size.

\begin{equation}
\frac{ {E}}{ {E_{OPT}}} = 1 
+\frac{(\frac{1}{\eta_{min}})\sum_{i=r_o+1}^{2r_o}\frac{w(c_i)}{r_o}}{(\frac{1}{
\eta_{max}})\sum_{i=1}^{r_o} \frac{w_o(c_i)}{r_o}}\\
\end{equation}

Let us denote the average size of jobs assigned to each machine by any
optimal algorithm by $X$, i.e., $X = \sum_{i=1}^{r_o}
\frac{w(c_i)}{r_o}$ and average size of jobs assigned by
Algorithm~\ref{energy_makespan_algo_indivisible} to $i=1 \ldots r_o$
and ${i=(r_o+1)} \ldots {2r_o}$ machines as $Y$ and $Z$ respectively.
\begin{equation}
\frac{ {E}}{ {E_{OPT}}} = 1 
+\frac{\eta_{max}Z}{\eta_{min}X}\\
\end{equation}
If $Z > \frac{X}{2}$, then $Y > \frac{X}{2}$. Since we have extra
$r_o$ machines, we can assign $Y$ to each machine in $i=1 \ldots
r_o$. But then the jobs will overflow the time constraint. This means
that optimally, $r_o$ machines will not be sufficient to accommodate
all the jobs which is a contradiction to our assumption that $r_o$ is
the optimal set of machines. Hence, $Z \leq \frac{X}{2}$.

\begin{equation}
\frac{ {E}}{ {E_{OPT}}} = 1 
+\frac{\eta_{max}\frac{X}{2}}{\eta_{min}X}\\
\end{equation}

\begin{equation}
 {E} =\Big(1 
+(\frac{1}{2})\frac{\eta_{max}}{\eta_{min}}\Big) {E_{OPT}}\\
\label{energy_bound}
\end{equation}
When all machines have same effectiveness ratio, equation \ref{energy_bound} 
reduces to the bound of $\frac{3}{2}$ for VSBP \cite{Friesen:1986}.
\end{proof}

\section{Conclusion}
\label{sec:con}

This paper provides a generic formulation of the problems of
scheduling in a system subject to pairs of constraint on various
resources available to its machines.  While not going into the details
of the specific machines and jobs run by them, the model simply
considers each machine to have a working power and an idle power
rating, which determines how much power the machine draws, while
working and while not.  The behavior of the system as a whole is
governed by the power ratings of its machines, and the need to run
jobs effectively.

A simple analysis shows that the problem of minimizing the makespan on
a power-constrained system is $\mathcal{NP}$-hard even with divisible
jobs.  This in turn implies that other interesting problems, and also
scheduling problems with non-divisible jobs, are $\mathcal{NP}$-hard
as well.  We get similar results with energy and makespan constraints
with non-divisible jobs.

We can modify the system model a little to get other interesting
problems. One such case could be a system model with multiple sources
of power, each of which has a different capacity and cost.  In such
systems, it is also of interest to minimize the overall cost incurred
for running a set of jobs.  This is easily seen to be
$\mathcal{NP}$-hard because if an efficient algorithm for minimal-cost
scheduling were to exist, then we could use the same for
energy-efficient scheduling by considering only one source supplying
power at cost 1 per unit energy.

Given the previous results, it also follows that many complex
objectives that mix two or more simpler objectives are likewise
intractable. The reason this so is that it is certainly no easier to
meet an objective when subjected to a constraint, than it is to meet
the objective without the constraint.

While this family of hardness results certainly puts paid to any hopes
of easy solutions to problems of scheduling for energy efficiency and
other desirable measures of effectiveness, it opens up some
interesting possibilities.  One direction we think would be fruitful
is to further analyze the sub-classes of problems which may permit
easy solutions; another would be construct suitable approximation
algorithms to achieve objectives within known bounds of the
optimal. We have pursued this particular avenue to some extent, and
suggest that there is much scope for further work along the same
lines.

\end{document}